\documentclass[letterpaper, 10 pt, conference]{ieeeconf}
\usepackage[utf8]{inputenc}

\IEEEoverridecommandlockouts

\usepackage{balance}

\usepackage{multirow}

\usepackage{algorithm}
\usepackage{tcolorbox}

\usepackage{amsmath,mathtools}
\usepackage{amsfonts,amssymb}
\renewcommand{\QED}{\QEDopen}

\usepackage{bbm}

\usepackage{xspace}

\usepackage{graphicx}
\graphicspath{{figs/}}

\usepackage{subcaption}
\usepackage{caption}

\usepackage{hyperref}
\hypersetup{colorlinks=true, linkcolor= blue, allcolors=blue, citecolor = blue, filecolor=black, urlcolor=blue}

\usepackage{todonotes}
\usepackage{comment}

\usepackage{tabularx}

\usepackage{diagbox}
\newcolumntype{K}[1]{>{\centering\arraybackslash}p{#1}}

\makeatletter
\newcommand{\multiline}[1]{%
  \begin{tabularx}{\dimexpr\linewidth-\ALG@thistlm}[t]{@{}X@{}}
    #1
  \end{tabularx}
}
\makeatother

\usepackage[per-mode=symbol]{siunitx}

\usepackage{cite}
\usepackage{array}
\usepackage{booktabs}

\makeatletter
\def\th@plain{%
  \itshape
}
\def\@begintheorem#1#2{\trivlist
   \item[\hskip\labelsep\bfseries #1\ #2.]}
\def\@opargbegintheorem#1#2#3{\trivlist
   \item[\hskip\labelsep\bfseries #1\ #2\ (#3).]}
\makeatother

\newtheorem{theorem}{Theorem}

\newtheorem{lemma}{Lemma}

\newtheorem{assumption}{Assumption}

\newcounter{remark}
\newenvironment{remark}{%
\par\vspace{3pt}\noindent\refstepcounter{remark}\textbf{Remark~\theremark:}}%
{\par\endtrivlist\unskip}

\newcounter{problem}
\newenvironment{problem}{%
\par\vspace{3pt}\noindent\refstepcounter{problem}\textbf{Problem~\theproblem:}}%
{\par\endtrivlist\unskip}

\usepackage[extramath,probability,reviewmark]{mylatexdefs}

\usepackage{stfloats}

\usepackage{hyperref}
\usepackage{graphicx}
\usepackage{amsmath,mathtools} 
\usepackage{mathrsfs}
\usepackage{amsfonts,amssymb}
\usepackage{upgreek}
\usepackage{bbm}
\usepackage{dsfont}
\usepackage{algorithm}
\usepackage{algpseudocode}

\usepackage{tikz}
\usetikzlibrary{arrows,positioning,plotmarks,arrows.meta}

\usepackage{pgfplots}
\usepgfplotslibrary{patchplots,groupplots}
\usepackage{grffile}

\pgfplotsset{compat=newest}  

\newlength{\figurewidth}
\newlength{\figureheight}

\pgfplotsset{
  every axis/.append style={
    label style={font=\large},
    tick label style={font=\normalsize},
    title style={font=\Large},
    legend style={font=\normalsize}
  }
}

\newcommand{\rostwo}{ROS\,2}
\newcommand{\limo}{\texttt{LIMO}}
\newcommand{\xrealdot}{\ensuremath{\dot{x}}}
\newcommand{\xreal}{\ensuremath{x}}
\newcommand{\xmodeldot}{\ensuremath{\dot{\hat{x}}}}
\newcommand{\xmodel}{\ensuremath{\hat{x}}}
\newcommand{\Real}{\ensuremath{\mathbb{R}}}
\newcommand{\statespace}{\ensuremath{\mathbb{R}^n}}
\newcommand{\controlspace}{\ensuremath{\mathcal{U}}}

\newcommand{\Jactual}{\ensuremath{J_{\mathrm{act}}}}
\newcommand{\Jmodel}{\ensuremath{J_{\mathrm{mod}}}}

\begin{document}

\title{\LARGE \bf Safety-Constrained Optimal Control and Hamiltonian-Gradient Learning for Systems with Unknown Dynamics}
\author{Panagiotis Kounatidis$^{1}$, {\IEEEmembership{Student Member, IEEE}}, and Andreas A. Malikopoulos$^{2}$, {\IEEEmembership{Senior Member, IEEE}}
\thanks{This research was supported in part by NSF under Grants CNS-2401007, CMMI-2348381, IIS-2415478, and in part by MathWorks.}
\thanks{$^{1}$Systems Engineering Program, Cornell University, Ithaca, NY, USA.}
\thanks{$^{2}$Andreas A. Malikopoulos is with the Applied Mathematics, Systems Engineering, Mechanical Engineering, Electrical \& Computer Engineering, and School of Civil \& Environmental Engineering, Cornell University, Ithaca, NY, USA.}
\thanks{Emails: {\tt\small \{pk586,amaliko\}@cornell.edu}.}
}
\maketitle

\begin{abstract}
We present a framework for safety-constrained optimal control when the plant dynamics are unknown and only an approximate model is available. We apply Pontryagin's Minimum Principle to a model-based problem whose cost penalizes model--plant state mismatch and establish conditions under which the constrained model and plant Hamiltonians have the same minimizer. These conditions show that optimality can be recovered by preserving control-relevant Hamiltonian structure rather than identifying the complete plant dynamics. Motivated by this observation, we develop a fully online method that uses integral trajectory data and policy iteration to learn the dynamics-dependent factors of the plant Hamiltonian gradient directly. Experiments on robotic cruise control demonstrate constrained model--plant control equivalence and online Hamiltonian-gradient reconstruction.
\end{abstract}

\section{Introduction}
Optimal control  seeks a control strategy that minimizes a performance criterion while satisfying the system dynamics and physical or safety constraints \cite{kirk2004}. Except for special classes such as linear systems with quadratic costs, obtaining such a strategy is generally difficult \cite{bertsekas2017}. Moreover, the main analytical tools---the calculus of variations, Pontryagin's Minimum Principle (PMP), and dynamic programming---traditionally presume accurate knowledge of the system dynamics. In applications such as autonomous driving, however, plant dynamics can be too complex or costly to model precisely. Approximate models are therefore unavoidable, and their mismatch with the plant can degrade both performance and closed-loop robustness \cite{sagmeister2024}. This raises two related questions: \emph{when does a control policy obtained from an approximate model remain optimal for the plant, and what must be learned online when the required structure is initially unknown?} The first calls for structural optimality conditions that are insensitive to modeling inacuracies \cite{malikopoulos2026}; the second calls for learning only the quantities that determine the optimal control law.

Reinforcement learning (RL) and approximate dynamic programming seek optimal policies directly from interaction data through policy search, value-function approximation, or actor--critic architectures \cite{sutton1999,recht2019,bertsekas1996,vamvoudakis2010,kiumarsi2018}. For continuous-time systems, online policy iteration and integral RL replace differential Bellman relations by identities over measured trajectory windows, enabling policy evaluation without full knowledge of the drift dynamics \cite{vamvoudakis2010,vrabie2009direct}. Recent work has extended data-driven policy iteration to unknown nonlinear systems, examined how numerical integration errors affect the learned controller, and introduced integral performance approximation using partial models and trajectory data \cite{chen2024homotopy,cao2024integral,wallace2025integral}. Further advances in integral RL improve excitation for more efficient learning in practical applications and relax the policy iteration's requirement for an initially admissible controller \cite{wallace2025continuous,zhang2026bias}. Complementary studies address RL-based optimal control under nonlinear uncertainty and the robustness of model-free LQR policy search \cite{garcia2024,li2026}. These methods motivate real-time learning from a single evolving trajectory, but typically target the value function, policy, or unknown dynamics broadly.
Owing to space limitations, we omit additional references and refer the reader to the thorough surveys in \cite{kiumarsi2018} and \cite{recht2019}.

A closely related framework integrates learning and optimal control through a penalized model-based problem that explicitly captures model--plant mismatch \cite{Malikopoulos2022a,Malikopoulos2024}, with an LQR implementation given in \cite{kounatidis2025combined}. More recently, \cite{malikopoulos2026} showed that exact model identification is unnecessary when learning preserves the structural properties that make model-based and plant-based decisions equivalent. This task-oriented viewpoint also appears in digital-twin learning, where models are refined for decision performance rather than prediction accuracy alone \cite{schena2024}.

This paper makes two contributions. First, we extend \cite{malikopoulos2026} to optimal control problems with safety constraints involving the state and control and establish conditions under which the constrained Hamiltonian minimizers of the model and plant coincide despite mismatched dynamics. Second, for smooth scalar affine systems with quadratic costs, we convert the resulting Hamiltonian-gradient matching condition into a fully online learning objective. An integral data equation and recursive least squares recover the control-relevant gradient factors, while policy improvement proceeds continuously without identifying the complete plant model or resetting its state. We validate both the constrained equivalence result and the learning method on robotic cruise-control experiments. The implementation code is available at \url{https://github.com/Panos20102k/Multi-Limo-Control}.

The paper is organized as follows.
Section~\ref{sec:problem-formulation} introduces the plant and penalized model-based problems and Section~\ref{sec:equivalence-results} establishes general and quadratic-effort equivalence results between these two problems. Section~\ref{sec:hamiltonian-gradient-learning} then develops the online Hamiltonian-gradient learning framework that learns from data the central condition of the equivalence results. Section~\ref{sec:experimental-results} provides the experimental implementation and results of our framework, and Section~\ref{sec:conclusions} concludes the paper.

\section{Problem Formulation}\label{sec:problem-formulation}
We consider the optimal control problem for a continuous dynamical system whose exact dynamics are unknown and is required to satisfy some safety constraints.

\subsection{Modeling framework}
The evolution of the plant is
\begin{equation}\label{eq:plant-dynamics}
\xrealdot(t)  = f(t, \xreal(t), u(t)), \quad \xreal(0) = x_0,
\end{equation}
and is constrained to satisfy 
\begin{equation}\label{eq:plant-constraint}
	c(t, \xreal(t), u(t)) \leq 0, \quad \mathrm{for\,\, all\,\,} t \in [0,T],
\end{equation}
where $\xreal(t) \in \statespace$, $u(t) \in \controlspace \subset \Real^m$ and $f : [0,T] \times \statespace \times \controlspace \rightarrow \statespace$ is the unknown plant dynamics map, which satisfies standard regularity conditions (e.g., Carath\'{e}odory conditions) such that for any admissible control $u(\cdot) \in \controlspace$, the system \eqref{eq:plant-dynamics} admits a unique absolutely continuous solution. The function $c : [0,T] \times \statespace \times \controlspace \rightarrow \Real$ is the plant constraint map.
\begin{remark}
	Considering $c$ to be a scalar function is not restrictive, as any $l$-vector constraint function ($l\leq m$) can be written compactly in one scalar function through the unit Heaviside step function \cite{kirk2004}.
\end{remark}
We consider that $\xreal(t)$ is fully observed for all $t \in [0, T]$.

The model of the plant that we have access to is given by
\begin{equation}\label{eq:model-dynamics}
\xmodeldot(t)  = \hat{f}(t, \xmodel(t), u(t)), \quad \xmodel(0) = x_0,
\end{equation}
subject to the same constraint map 
\begin{equation}\label{eq:model-constraint}
	c(t, \xmodel(t), u(t)) \leq 0, \quad \mathrm{for\,\, all\,\,} t \in [0,T],
\end{equation}
where $\xmodel(t) \in \statespace$ and $\hat{f} : [0,T] \times \statespace \times \controlspace \rightarrow \statespace$ is the available model dynamics map.
\begin{remark}
	The model and the plant share the same initial condition $x_0$ and are driven by the same control input $u(\cdot)$. The model state $\xmodel(t)$ is available at all times.
\end{remark}

\subsection{Original optimal control problem for the plant}
The performance of the plant is evaluated through
\begin{equation}\label{eq:plant-cost-functional}
	\Jactual = \int_{0}^{T} \ell(t, \xreal(t), u(t))\,dt + \phi(\xreal(T)),
\end{equation}
where $\ell : [0,T] \times \statespace \times \controlspace \rightarrow \Real$ is the running cost and $\phi : \statespace \rightarrow \Real$ the terminal cost. The problem we want to address is given as follows: 
\begin{problem}\label{prob:plant-optimal}
	Minimize \Jactual\, over $u(\cdot) \in \controlspace$, subject to \eqref{eq:plant-dynamics} and \eqref{eq:plant-constraint}.
\end{problem}
However, the plant dynamics $f$ are unknown, so Problem~\ref{prob:plant-optimal} cannot be solved directly.

\subsection{Model-based surrogate problem with penalized cost}
To overcome the lack of knowledge of $f$, we construct a surrogate optimal control problem based on the model dynamics \eqref{eq:model-dynamics}. The key idea is to augment the running cost with a penalty term that quantifies the discrepancy between the model state and the observed plant state. To this end, define
\begin{align}\label{eq:model-cost-functional}
	\Jmodel \;&=\; \int_{0}^{T} \ell(t, \xmodel(t), u(t)) + \beta(t)||\xmodel(t) - \xreal(t)||^2\,dt \\ \nonumber
	&+ \phi(\xmodel(T)),
\end{align}
where $\beta : [0, T] \rightarrow \Real$ is a given time-varying weighting function. Then, we consider the following problem.
\begin{problem}\label{prob:model-based}
	Minimize \Jmodel\, over $u(\cdot) \in \controlspace$, subject to \eqref{eq:model-dynamics} and \eqref{eq:model-constraint}.
\end{problem}

\subsection{Hamiltonian Analysis and Optimality Conditions}\label{sec:hamiltonian-analysis}
We apply PMP to the model-based penalized problem, assuming throughout the required regularity conditions and suppressing time arguments when clear. Its Hamiltonian is
\begin{align}
\hat H(t,\xmodel,\xreal,u,\hat\lambda,\hat\mu)
&= \ell(t,\xmodel,u)+\hat\lambda^{\mathrm T}\hat f(t,\xmodel,u)
+\beta\|\xmodel-\xreal\|^2 \nonumber\\
&\quad+\hat\mu c(t,\xmodel,u).\label{eq:model-hamiltonian}
\end{align}
If $u^\circ(\cdot)$ is optimal for Problem~\ref{prob:model-based}, with state $\xmodel^\circ(\cdot)$, then there exist a continuous costate $\hat\lambda^\circ(\cdot)$ and a multiplier $\hat\mu^\circ(\cdot)\geq0$ such that, almost everywhere,
\begin{align}
\xmodeldot^\circ&=\hat f(t,\xmodel^\circ,u^\circ),\qquad \xmodel^\circ(0)=x_0,\label{eq:pmp-model-state}\\
\dot{\hat\lambda}^\circ&=-\nabla_{\xmodel}\hat H(t,\xmodel^\circ,\xreal, u^\circ,\hat\lambda^\circ,\hat\mu^\circ),\label{eq:pmp-model-costate}\\
u^\circ&\in\operatorname*{arg\,min}_{u\in\controlspace}
\hat H(t,\xmodel^\circ,\xreal,u,\hat\lambda^\circ,\hat\mu^\circ),\label{eq:pmp-model-control}
\end{align}
with $\hat\lambda^\circ(T)=\nabla_{\xmodel}\phi(\xmodel^\circ(T))$ and the usual complementary-slackness conditions, $\hat\mu^\circ c=0$, \cite{bryson1975applied}. The original plant problem has the same PMP structure after removing the penalty term and replacing all estimated (hatted) objects by their true, unhatted counterparts.

\begin{remark}\label{remark:differentiation}
If $c(t,\xmodel)$ does not depend explicitly on $u$, it is differentiated until $c^{(q)}(t,\xmodel,u)$ first depends on $u$. Then $c^{(q)}$ replaces $c$ in the Hamiltonian conditions, and an active constraint must also satisfy
\begin{equation}\label{eq:pmp-model-tangency}
\hat N(\xmodel,t)\doteq[c\;\dot{c}\;\cdots\;c^{(q-1)}]^{\mathrm T}=0.
\end{equation}
\end{remark}

\subsection{Constrained Hamiltonian minimization. Existence and uniqueness}
We first state compact conditions ensuring that the pointwise Hamiltonian minimizations are well posed.

\begin{assumption}\label{assumption:general-control-set}
The admissible input set $\mathcal U\subset\Real^m$ is nonempty, closed, and convex, but need not be bounded.
\end{assumption}
\begin{assumption}\label{assumption:general-hamiltonians}
For a.e. $t$ and all relevant $(x,\hat x,\lambda,\hat\lambda,\mu,\hat\mu)$, the maps
\[
u\mapsto H(t,x,u,\lambda,\mu),\qquad
u\mapsto\hat H(t,\hat x,x,u,\hat\lambda,\hat\mu)
\]
are proper, lower semicontinuous, and convex on $\mathcal U$, and
\begin{align*}
\|u\|&\to\infty,\ u\in\mathcal U \Longrightarrow \\
H(t,x,u,\lambda,\mu)&\to+\infty, \quad \mathrm{and} \quad  \hat H(t,\hat x,x,u,\hat\lambda,\hat\mu)\to+\infty.
\end{align*}
\end{assumption}

\begin{theorem}\label{theorem:existence-uniqueness}
{\itshape
Suppose Assumptions~\ref{assumption:general-control-set}--\ref{assumption:general-hamiltonians} hold. Then, for almost every $t\in[0,T]$, the sets
\[
\arg\min_{u\in\mathcal U}H(t,x,u,\lambda,\mu),
\qquad
\arg\min_{u\in\mathcal U}\hat H(t,\hat x,x,u,\hat\lambda,\hat\mu)
\]
are nonempty, closed, and convex. If both Hamiltonians are strictly convex in $u$ on $\mathcal U$, their minimizers are unique almost everywhere.
}
\end{theorem}
\begin{proof}
See \cite[Theorem~1]{malikopoulos2026}.
\end{proof}

\section{Equivalence Results}\label{sec:equivalence-results}
For the convex-analysis preliminaries used below, including subdifferentials, normal cones, and constrained first-order optimality conditions, we refer the reader to \cite{malikopoulos2026}.

\subsection{General Equivalence Results}
The following equivalence results are pointwise in time
and rely on the structure of the instantaneous Hamiltonian
minimization problems induced by the two optimal control
formulations. We provide two complementary equivalences
results. The first is stated in a convex-analysis form (subdifferentials and normal cones) and accommodates nonsmooth costs, while the second  specializes in a commonly used structural setting (quadratic control effort and mild growth conditions), which
yields simple, verifiable conditions for existence, uniqueness,
and equivalence.

\begin{theorem}\label{theorem:two}
{\itshape
Suppose Assumptions~\ref{assumption:general-control-set}--\ref{assumption:general-hamiltonians} hold and define the full feasible set
\begin{equation}\label{eq:full-feasible-set}
\mathcal C\doteq\{u\in\mathcal U:c(\cdot,u)\leq0\},
\end{equation}
where $u\mapsto c(\cdot,u)$ is proper, lower semicontinuous, and convex, and $\mathcal C$ is nonempty. Fix a time at which both Hamiltonian minimizations are well posed. If some $\bar u\in\mathcal C$ satisfies
\begin{equation}\label{eq:switching_gradient_match}
\partial_u H(t,x,\bar u,\lambda,\mu)
=\partial_u\hat H(t,\hat x,x,\bar u,\hat\lambda,\hat\mu),
\end{equation}
then
\begin{align}
0&\in\partial_uH(t,x,\bar u,\lambda,\mu)+N_{\mathcal C}(\bar u)\nonumber \Longleftrightarrow \\
0&\in\partial_u\hat H(t,\hat x,x,\bar u,\hat\lambda,\hat\mu)+N_{\mathcal C}(\bar u).
\label{eq:full-feasible-optimality}
\end{align}
Thus, if either inclusion holds, $\bar u$ minimizes both Hamiltonians over $\mathcal C$. If both are strictly convex in $u$ on $\mathcal C$, their unique minimizers coincide.
}
\end{theorem}
\begin{proof}
Assumption~\ref{assumption:general-control-set} and the hypotheses on $c$ imply that $\mathcal C$ is nonempty, closed, and convex. Assumption~\ref{assumption:general-hamiltonians} implies existence of both minima on $\mathcal C$. For any proper, lower semicontinuous convex function $\Psi : \mathcal{U} \to \Real \cup \{\infty\}$,
\[
\bar u\in\arg\min_{u\in\mathcal C}\Psi(u)
\Longleftrightarrow 0\in\partial\Psi(\bar u)+N_{\mathcal C}(\bar u).
\]
Writing $\Psi=H$ and $\Psi=\hat H$, respectively, and using \eqref{eq:switching_gradient_match} yields
\[
\partial_uH(\bar u)+N_{\mathcal C}(\bar u)
=\partial_u\hat H(\bar u)+N_{\mathcal C}(\bar u),
\]
which proves \eqref{eq:full-feasible-optimality} and the common-minimizer claim. Strict convexity on $\mathcal C$ gives uniqueness.
\end{proof}

\setcounter{remark}{4}
\begin{remark}
Condition~\eqref{eq:switching_gradient_match} is a condition that can be learned; we do so in Section~\ref{sec:hamiltonian-gradient-learning}.
\end{remark}

\subsection{Specialization to Quadratic Control Effort}
The preceding hypotheses become readily verifiable for quadratic control effort.

\begin{assumption}\label{assumption:quadratic-effort}
$\mathcal U$ is nonempty, closed, and convex, and
$\ell(t,z,u)=\ell_0(t,z)+\tfrac12u^{\mathrm T}R(t)u$, where $\ell_0$ is continuous and $R(t)\succeq r_{\min}I$ for some $r_{\min}>0$.
\end{assumption}
\begin{assumption}\label{assumption:linear-growth}
For almost every $t$, the control-dependent terms $u\mapsto\lambda^{\mathrm T}f$, $u\mapsto\hat\lambda^{\mathrm T}\hat f$, $u\mapsto\mu c$, and $u\mapsto\hat\mu c$ are convex on $\mathcal U$ and bounded in absolute value by locally bounded coefficients times $1+\|u\|$.
\end{assumption}

\begin{lemma}\label{lemma:unique-minimizers}
Under Assumptions~\ref{assumption:quadratic-effort}--\ref{assumption:linear-growth}, the pointwise plant and model Hamiltonian minimizations over $\mathcal U$ admit unique minimizers for almost every $t$.
\end{lemma}
\begin{proof}
Fix $t\in[0,T]$ such that Assumptions~\ref{assumption:quadratic-effort}--\ref{assumption:linear-growth} hold (this is the case for almost every $t$). We prove the claim for the model Hamiltonian. The proof for the plant Hamiltonian is identical. For fixed $(t,\hat x,x,\hat\lambda)$, we define the function
\begin{align*}
\Psi(u)&\doteq\hat H(t,\hat x,x,u,\hat\lambda,\hat\mu)\\
&=\ell_0(t,\hat x)+\tfrac12u^{\mathrm T}R(t)u
+\beta(t)\|\hat x-x\|^2\\
&\quad+\hat\lambda^{\mathrm T}\hat f(t,\hat x,u)
+\hat\mu c(t,\hat x,u).
\end{align*}
By Assumption~\ref{assumption:linear-growth}, the maps
$u\mapsto\hat\lambda^{\mathrm T}\hat f(t,\hat x,u)$ and
$u\mapsto\hat\mu c(t,\hat x,u)$ are convex on $\mathcal U$.
Therefore, $\Psi$ is convex on $\mathcal U$. Moreover, since
$R(t)\succeq r_{\min}I$, the quadratic term is $r_{\min}$-strongly
convex on $\mathcal U$, hence $\Psi$ is strongly convex on
$\mathcal U$ as the sum of a strongly convex function and two convex
functions.

Let $r_{\max}(t)\doteq\|R(t)\|$ and note that
\[
\frac12u^{\mathrm T}R(t)u\geq\frac{r_{\min}}2\|u\|^2,
\qquad\forall u\in\mathcal U.
\]
By Assumption~\ref{assumption:linear-growth}, denote the corresponding
locally bounded linear-growth coefficients by
$c_{\hat f}(t,\hat x)$ and $c_{c}(t,\hat x)$. Then
\begin{align*}
\Psi(u)\geq{}&\ell_0(t,\hat x)+\beta(t)\|\hat x-x\|^2
+\frac{r_{\min}}2\|u\|^2\\
&-\big(c_{\hat f}(t,\hat x)+c_{c}(t,\hat x)\big)(1+\|u\|).
\end{align*}
The right-hand side is a quadratic function of $\|u\|$ with positive
leading coefficient $r_{\min}/2$; therefore,
\[
\|u\|\to\infty,\quad u\in\mathcal U
\quad\Longrightarrow\quad\Psi(u)\to+\infty.
\]
Namely, $\Psi$ is coercive on $\mathcal U$.

Let $m^\star\doteq\inf_{u\in\mathcal U}\Psi(u)$ and let
$\{u_k\}\subset\mathcal U$ be a minimizing sequence with
$\Psi(u_k)\downarrow m^\star$. Coercivity implies that $\{u_k\}$ is
bounded. Otherwise, $\|u_k\|\to\infty$ along a subsequence would force
$\Psi(u_k)\to+\infty$, contradicting
$\Psi(u_k)\downarrow m^\star<+\infty$.

Since $\{u_k\}$ is bounded, there exists a subsequence (not relabeled)
and $\bar u\in\Real^m$ such that $u_k\to\bar u$. Because $\mathcal U$
is closed, $\bar u\in\mathcal U$. Finally, $\Psi$ is convex (hence
continuous on the relative interior of $\mathcal U$) and, under the
present assumptions, lower semicontinuous on $\mathcal U$; thus
\[
\Psi(\bar u)\leq\liminf_{k\to\infty}\Psi(u_k)
=m^\star.
\]
This yields $\Psi(\bar u)=m^\star$. Therefore, $\bar u$ is a minimizer
and the argmin set is nonempty.

Because $\Psi$ is strongly convex on $\mathcal U$, it admits at most
one minimizer on $\mathcal U$. Indeed, if $u_1\neq u_2$ were both
minimizers, then for any $\theta\in(0,1)$ strong convexity would imply
\[
\Psi(\theta u_1+(1-\theta)u_2)
<\theta\Psi(u_1)+(1-\theta)\Psi(u_2)=m^\star,
\]
a contradiction. Hence, the minimizer is unique.
\end{proof}

\begin{theorem}\label{theorem:quadratic-effort}
{\itshape
Suppose Assumptions~\ref{assumption:quadratic-effort}--\ref{assumption:linear-growth} hold, and let the plant and model PMP trajectories be $(x^*,\lambda^*,\mu^*,u^*)$ and $(\hat x^\circ,\hat\lambda^\circ,\hat\mu^\circ,u^\circ)$, respectively. If, almost everywhere,
\begin{align}
\nabla_u H(t,x^*,u^\circ,\lambda^*,\mu^*)
&=\nabla_u\hat H(t,\hat x^\circ,x^*,u^\circ,\hat\lambda^\circ,\hat\mu^\circ),\label{eq:quadratic-gradient-match}\\
x^*(t)&=\hat x^\circ(t),\label{eq:state_alignment}
\end{align}
then $u^*(t)=u^\circ(t)$ almost everywhere. A sufficient set of verifiable conditions for \eqref{eq:quadratic-gradient-match} is
\begin{align}\nonumber
\nabla_u f(t,x^*(t),u)
&=\nabla_u\hat f(t,\hat x^\circ(t),u),\quad \forall u\in\mathcal U,\ \mathrm{a.e. } \quad t,\\ \nonumber
\lambda^*(t)&=\hat\lambda^\circ(t) \quad \mathrm{a.e.}.,\\
\mu^*(t)&=\hat\mu^\circ(t) \quad \mathrm{a.e.}.
\label{eq:sufficient-gradient-match}
\end{align}
}
\end{theorem}
\begin{proof}
Fix a time at which the hypotheses hold and suppress its dependence. Define
\begin{align*}
\Psi_{\rm p}(u)&=H(x^*,u,\lambda^*,\mu^*),\\[-2pt]
\Psi_{\rm m}(u)&=\hat H(\hat x^\circ,x^*,u,\hat\lambda^\circ,\hat\mu^\circ).
\end{align*}
Lemma~\ref{lemma:unique-minimizers} gives unique minimizers $u^*$ and $u^\circ$. Convex differentiability yields
\begin{align*}
u^\circ=\arg\min_{u\in\mathcal U}\Psi_{\rm m}(u)
\quad\Longleftrightarrow\quad
\langle\nabla\Psi_{\rm m}(u^\circ),v-\hat u^\circ\rangle&\geq0,\\[-2pt]
&\hspace{-18mm}\forall v\in\mathcal U.
\end{align*}
By \eqref{eq:quadratic-gradient-match},
$\nabla\Psi_{\rm p}(u^\circ)=\nabla\Psi_{\rm m}(u^\circ)$; hence the same inequality characterizes $u^\circ$ as a minimizer of $\Psi_{\rm p}$. Uniqueness gives $u^*=u^\circ$. Since this holds for a.e. $t$, the claimed equality follows. Under \eqref{eq:state_alignment}, substituting
\eqref{eq:sufficient-gradient-match} into the two Hamiltonian gradients yields \eqref{eq:quadratic-gradient-match}.
\end{proof}

\section{Hamiltonian Gradient Learning}\label{sec:hamiltonian-gradient-learning}
The equivalence condition in \eqref{eq:switching_gradient_match} suggests a
learning objective that is weaker than identifying the complete plant
dynamics: learn directly the plant-side quantity required by the Hamiltonian
subdifferential matching condition. We develop this idea for smooth,
unconstrained Hamiltonians, for which subdifferential matching reduces to
Hamiltonian-gradient matching \eqref{eq:quadratic-gradient-match}. Our scope is scalar affine systems with
quadratic costs, and we demonstrate the effectiveness of the resulting
algorithm on a robotic cruise control application in Section~\ref{sec:learning-cruise-control}.

\subsection{Problem setup and gradient-matching objective}
Consider the unknown scalar affine plant and a known approximate model
\begin{equation}\label{eq:hgl-dynamics}
\begin{aligned}
    \dot{x}&=ax+bu+d, &x(0)&=x_0,\\
    \dot{\hat{x}}&=\hat{a}\hat{x}+\hat{b}u+\hat{d},
    &\hat{x}(0)&=x_0,
\end{aligned}
\end{equation}
where $a,b,d\in\Real$ are unknown and
$\hat{a},\hat{b},\hat{d}\in\Real$ are known estimates. For cruise control,
\begin{equation}\label{eq:hgl-cruise-control-map}
    x=v-v_{\mathrm{ref}},\qquad
    a=-\frac{1}{\tau},\qquad b=\frac{k}{\tau},\qquad
    d=av_{\mathrm{ref}}.
\end{equation}
Here, $x$ is the velocity-tracking error, $\tau>0$ is the actuation time
constant, $k>0$ is the steady-state actuation gain, and $d$ is the constant
offset induced by expressing the velocity dynamics about a nonzero reference.
The hatted parameters have the analogous meanings for the model.

The discounted plant and penalized model costs are
\begin{align}
    J&=\int_0^\infty e^{-\rho t}(qx^2+ru^2)\,dt,
    \label{eq:hgl-plant-cost}\\
    \hat{J}&=\int_0^\infty e^{-\rho t}
    [q\hat{x}^2+ru^2+\beta\lVert\hat{x}-x\rVert^2]dt,
    \label{eq:hgl-model-cost}
\end{align}
where $q,r,\beta,\rho>0$. Discounting is essential in the cruise control
application because maintaining a nonzero reference generally requires the
nonzero steady-state input $u_{\mathrm{ss}}=v_{\mathrm{ref}}/k$; without
discounting, the integral of $ru_{\mathrm{ss}}^2$ over an infinite horizon
would diverge.

The plant and model Hamiltonians are
\begin{align*}
    H(x,u,\lambda)
    &=qx^2+ru^2+\lambda(ax+bu+d),\\
    \hat{H}(\hat{x},x,u,\hat{\lambda})
    &=q\hat{x}^2+ru^2+\beta\lVert\hat{x}-x\rVert^2
      +\hat{\lambda}(\hat{a}\hat{x}+\hat{b}u+\hat{d}).
\end{align*}
The smooth counterpart of the subdifferential matching condition of Theorem \ref{theorem:two} is
\begin{equation}\label{eq:hgl-gradient-match}
    \nabla_uH(x,u,\lambda)
    =\nabla_u\hat{H}(\hat{x},x,u,\hat{\lambda}).
\end{equation}
The model-side gradient is directly computable from the known model, whereas the plant-side gradient cannot be evaluated directly because $a$, $b$, and $d$ are unknown. Our goal is therefore to recover from plant trajectories only the unknown quantity needed to evaluate the left-hand side of \eqref{eq:hgl-gradient-match}.

\subsection{Policy evaluation from online data}
For policy iteration index $i$, consider the affine policy and
quadratic-affine cost-to-go
\begin{equation}\label{eq:hgl-parameterizations}
    \mu_i(x)=k_{1,i}x+k_{0,i},\qquad
    J_i(x)=P_i x^2+2s_i x+c_i.
\end{equation}
The affine policy accommodates the constant term $d$, while the
quadratic-affine $J_i$ is the exact cost-to-go structure associated with any
stabilizing affine policy for the scalar affine LQ problem. The corresponding costate is
$\lambda_i = J_i'(x)=2P_i x+2s_i$.

For a fixed policy, the discounted policy-evaluation equation is
\begin{equation}\label{eq:hgl-fixed-policy-relation}
    J_i'(x)[ax+b\mu_i(x)+d]
    =\rho J_i(x)-qx^2-r\mu_i^2(x).
\end{equation}

To obtain sufficient information from the plant, the applied input during learning is
\begin{equation}\label{eq:hgl-exploration-control}
    u=\mu_i(x)+\delta u,
\end{equation}
where $\delta u$ is bounded exploration. Under this input, the closed-loop dynamics become
\begin{equation}\label{eq:hgl:closed-loop-dynamics}
    \dot{x}=ax+b(\mu_i(x)+\delta u)+d.
\end{equation}
Define the control-relevant factor
\begin{equation}\label{eq:hgl-G-definition}
\begin{aligned}
    G_i(x)&\doteq bJ_i'(x)=\gamma_{1,i}x+\gamma_{0,i},\\
    \gamma_{1,i}&=2bP_i,\qquad \gamma_{0,i}=2bs_i.
\end{aligned}
\end{equation}
Using \eqref{eq:hgl-fixed-policy-relation} and \eqref{eq:hgl:closed-loop-dynamics}, the time evolution of the
cost-to-go function along the exploratory plant trajectory becomes
\begin{equation}\label{eq:hgl-cost-evolution}
    \dot{J}_i(x)=\rho J_i(x)-qx^2-r\mu_i^2(x)+G_i(x)\delta u.
\end{equation}

For the sample window $[t_k,t_{k+1}]$ of length $\Delta$, let
$\sigma=t-t_k$, $\zeta=e^{-\rho\Delta}$, and define
\begin{align*}
    I_{\ell,k}&\doteq\int_{t_k}^{t_{k+1}}e^{-\rho\sigma}
    [qx^2+r\mu_i^2(x)]dt,\nonumber\\
    I_{x\delta u,k}&\doteq\int_{t_k}^{t_{k+1}}e^{-\rho\sigma}
    x\delta u\,dt,\qquad
    I_{\delta u,k}\doteq\int_{t_k}^{t_{k+1}}e^{-\rho\sigma}
    \delta u\,dt.\label{eq:hgl-window-integrals}
\end{align*}
Integrating \eqref{eq:hgl-cost-evolution} gives the
compact data equation
\begin{equation*}\label{eq:hgl-integral-data-equation}
\begin{split}
I_{\ell,k}={}&P_i(x_k^2-\zeta x_{k+1}^2)
 +2s_i(x_k-\zeta x_{k+1})+c_i(1-\zeta)\\
&+\gamma_{1,i}I_{x\delta u,k}
 +\gamma_{0,i}I_{\delta u,k}.
\end{split}
\end{equation*}
Thus, for $k\in\{1,\ldots,N\}$, where $N$ is the total number of sampled
time windows, define
\begin{align}
    y_k&\doteq I_{\ell,k},\nonumber\\
    \phi_k&\doteq
    [x_k^2-\zeta x_{k+1}^2,\;2(x_k-\zeta x_{k+1}),\;1-\zeta,
    \;I_{x\delta u,k},\;I_{\delta u,k}]^{\mathrm T},\nonumber\\
    \theta_i&\doteq
    [P_i,\;s_i,\;c_i,\;\gamma_{1,i},\;\gamma_{0,i}]^{\mathrm T}.
    \label{eq:hgl-regression-definitions}
\end{align}
Then
\begin{equation}\label{eq:hgl-regression}
    y_k=\phi_k^{\mathrm T}\theta_i,
\end{equation}
which is suitable for online recursive least-squares (RLS) estimation. The standard
estimator equations are omitted for brevity.

\subsection{Policy improvement and fully online algorithm}
Since the plant Hamiltonian gradient under policy $i$ is
\begin{equation}\label{eq:hgl-plant-gradient}
    \nabla_uH=2ru+G_i(x)=2ru+\gamma_{1,i}x+\gamma_{0,i},
\end{equation}
the stationarity condition yields the improved policy
\begin{equation}\label{eq:hgl-policy-update}
\begin{aligned}
    \mu_{i+1}(x)
    &=-\frac{\gamma_{1,i}^{\mathrm{est}}x+
    \gamma_{0,i}^{\mathrm{est}}}{2r},\\
    k_{1,i+1}&=-\frac{\gamma_{1,i}^{\mathrm{est}}}{2r},
    \qquad
    k_{0,i+1}=-\frac{\gamma_{0,i}^{\mathrm{est}}}{2r},
\end{aligned}
\end{equation}
where the superscript $\mathrm{est}$ denotes the latest available RLS estimate. The learner therefore estimates the dynamics-dependent factors
$\gamma_{1,i}=2bP_i$ and $\gamma_{0,i}=2bs_i$ that determine the plant
Hamiltonian gradient, rather than identifying $a$, $b$, and $d$ separately.
The initial policy must stabilize the plant, and the exploration signal must
make \eqref{eq:hgl-regression} informative. Subject to these requirements,
policy evaluation and improvement proceed continuously without resetting the
physical state.

Along the measured behavior trajectory, the recovered plant-side gradient is
\begin{equation}\label{eq:hgl-learned-gradient}
    \nabla_uH_{\mathrm{learned}}(t)
    =2ru(t)+\gamma_1^{\mathrm{est}}(t)x(t)
    +\gamma_0^{\mathrm{est}}(t).
\end{equation}
Let $J^*(x)=P^*x^2+2s^*x+c^*$ denote the cost-to-go following the optimal policy. Its coefficients are computed analytically from the Hamilton-Jacobi-Bellman equation. Then, the validation model signal is
\begin{equation}\label{eq:hgl-model-gradient}
    \nabla_u H_{\mathrm{model}}(t)
    =2ru(t)+\gamma_1^* x(t)+\gamma_0^*.
\end{equation}
where $\gamma_1^* = 2 b P^*$ and $\gamma_0^* = 2 b s^*$. At every simulation step, the reference gradient is evaluated at the same state and actual applied input as the learned gradient. Consequently, their difference is 
\begin{equation*}
    \nabla_uH_{\mathrm{learned}} - \nabla_u H_{\mathrm{model}} = (\gamma_1^{\mathrm{est}} - \gamma_1^*)x + \gamma_0^{\mathrm{est}} - \gamma_0^*
\end{equation*}
Therefore, the comparison specifically tests whether the learner has recovered the state-dependent part of the true Hamiltonian gradient. The whole Hamiltonian gradient learning algorithm is summarized in Algorithm~\ref{alg:hgl-online}.

\begin{algorithm}[H]
\caption{Fully online Hamiltonian-gradient learning}
\label{alg:hgl-online}
\begin{enumerate}
    \item Initialize a stabilizing affine policy
    $\mu_0(x)=k_{1,0}x+k_{0,0}$, the least-squares estimate, and its
    covariance matrix.
    \item For the current policy $\mu_i$, apply
    $u=\mu_i(x)+\delta u$. Over each of the $N$ data windows, compute the
    components of $\phi_k$ and $y_k$ in
    \eqref{eq:hgl-regression-definitions} and update the estimate of
    $\theta_i$ from \eqref{eq:hgl-regression}.
    \item After sufficiently informative data have been collected, update
    $(k_{1,i+1},k_{0,i+1})$ using \eqref{eq:hgl-policy-update}, and continue without state reseting until policy coefficients converge.
    \item Remove exploration and exploit the converged policy. Evaluate
    gradient matching by comparing \eqref{eq:hgl-learned-gradient} with
    \eqref{eq:hgl-model-gradient}.
\end{enumerate}
\end{algorithm}

\section{Experimental Results}\label{sec:experimental-results}
In this section, we apply our framework in two robotic applications. In the first we demonstrate the emergence of equivalence of optimal controls between the model-based and original plant problem through a safety-constrained cruise control application. In the second, we apply the Hamiltonian gradient learning framework to learn the matching condition \eqref{eq:quadratic-gradient-match} online and optimally perform unconstrained cruise control. The experiments are conducted on the \limo\, robots \cite{limo} using the \rostwo\, system \cite{ros2}. 

\subsection{Safety-Constrained Cruise Control}\label{sec:cruise-control}
We consider an ego \limo\, that we control, and is constrained to maintain a safety distance from another \limo\, in the front that we do not control.

\subsubsection{Setup}
Let $x=[p\;v]^{\mathrm T} \in \Real^2$ be the plant state of the ego \limo, with $u\in\controlspace=[u_{\min},u_{\max}] \in \Real^2$. The plant dynamics are
\begin{equation}\label{eq:cc-plant-dynamics}
\dot x=
\begin{bmatrix}0&1\\0&-1/\tau\end{bmatrix}x
+\begin{bmatrix}0\\k/\tau\end{bmatrix}u, \qquad x(0)=x_0,
\end{equation}
where $k>0$ is the actuation gain and $\tau>0$ the delay time constant, subject to
\begin{equation}\label{eq:cc-plant-constraint}
c(x)\doteq\delta-\xi(p_{\mathrm f}-p)\leq0,
\end{equation}
where $p_{\mathrm f} \in \Real$ is the position of the front \limo, $\delta >0$ is the safety distance and $\xi>0$ is the reaction time coefficient. The model has the same dynamics and constraint, with state $\hat x=[\hat p\;\hat v]^{\mathrm T}$ and parameter estimates $\hat k$ and $\hat\tau$ in place of $k$ and $\tau$.

The plant performance is
\begin{equation}\label{eq:cc-plant-cost-functional}
J=\int_0^T\!\left[q(v-v_{\mathrm{ref}})^2+ra^2\right]dt
+h[v(T)-v_{\mathrm{ref}}]^2,
\end{equation}
where $q,r,h>0$ and
\begin{equation}\label{eq:cc-plant-acceleration}
a\doteq\frac{ku-v}{\tau}=\dot v.
\end{equation}

\subsubsection{Hamiltonian Minimization and Control Laws}
For the model-based problem,
\begin{align}
\hat H={}&q(\hat v-v_{\mathrm{ref}})^2+r\hat a^2
+\hat\lambda_1\hat v+\hat\lambda_2\hat a+\hat\mu c \nonumber\\
&+\beta_1(\hat p-p)^2+\beta_2(\hat v-v)^2,\label{eq:cc-model-hamiltonian}
\end{align}
where $\hat a=(\hat k u-\hat v)/\hat\tau$ and $c$ denotes the first derivative of the safety constraint that depends explicitly on $u$ (see Remark~\ref{remark:differentiation}). Since the penalty is independent of $u$, it does not enter the pointwise minimization. Under state alignment, application of the PMP conditons \eqref{eq:pmp-model-state}--\eqref{eq:pmp-model-control} yields
\begin{equation}
\ddot{\hat v}=\frac qr(\hat v-v_{\mathrm{ref}}).
\end{equation}
Let $\omega=\sqrt{q/r}$, $\hat w_0=\hat v(0)-v_{\mathrm{ref}}$, and
\begin{equation}\label{eq:cc-B}
\hat B=-\hat w_0
\frac{(h/r)\cosh(\omega T)+\omega\sinh(\omega T)}
{\omega\cosh(\omega T)+(h/r)\sinh(\omega T)}.
\end{equation}
Before the safety constraint activates at $t_s$, the model-based unconstrained solution is
\begin{align}
\hat v^*(t)&=v_{\mathrm{ref}}+\hat w_0\cosh(\omega t)+\hat B\sinh(\omega t),\label{eq:cc-model-optimal-velocity}\\
u^*(t)&=\frac{\hat v^*(t)+\hat\tau\omega[\hat w_0\sinh(\omega t)+\hat B\cosh(\omega t)]}{\hat k}.
\label{eq:cc-model-optimal-control}
\end{align}
On the active constraint, the tangency conditions \eqref{eq:pmp-model-tangency} give
\begin{equation}\label{eq:cc-model-optimal-control-boundary}
u^*(t)=\frac{\dot p_{\mathrm f}(t)+\hat\tau\ddot p_{\mathrm f}(t)}{\hat k}.
\end{equation}
Thus, the final control strategy of the model-based surrogate problem is the projection of \eqref{eq:cc-model-optimal-control} and \eqref{eq:cc-model-optimal-control-boundary} to the one-dimensional simplex line, i.e.,
\begin{equation}\label{eq:cc-model-final-control}
u_{\mathrm{mb}}(t)=\Pi_{[u_{\min},u_{\max}]}[u^*(t)].
\end{equation}
Applying the same procedure to the plant gives
\begin{equation}\label{eq:cc-plant-final-control}
u_{\mathrm{opt}}(t)=\Pi_{[u_{\min},u_{\max}]}[u^*(t)],
\end{equation}
where $u^*(t)$ is given by \eqref{eq:cc-model-optimal-control} and \eqref{eq:cc-model-optimal-control-boundary} but with the real parameters substituted in these expressions instead of their estimates.
\begin{figure*}[t]
    \centering

    \begin{subfigure}[t]{0.32\textwidth}
        \centering
        \resizebox{\textwidth}{!}{\input{figures/position-final.tex}}
        \caption{Position over time of front and ego \limo}
        \label{fig:state_equiv}
    \end{subfigure}
    \hfill
    \begin{subfigure}[t]{0.32\textwidth}
        \centering
        \includegraphics[width=\textwidth,keepaspectratio]{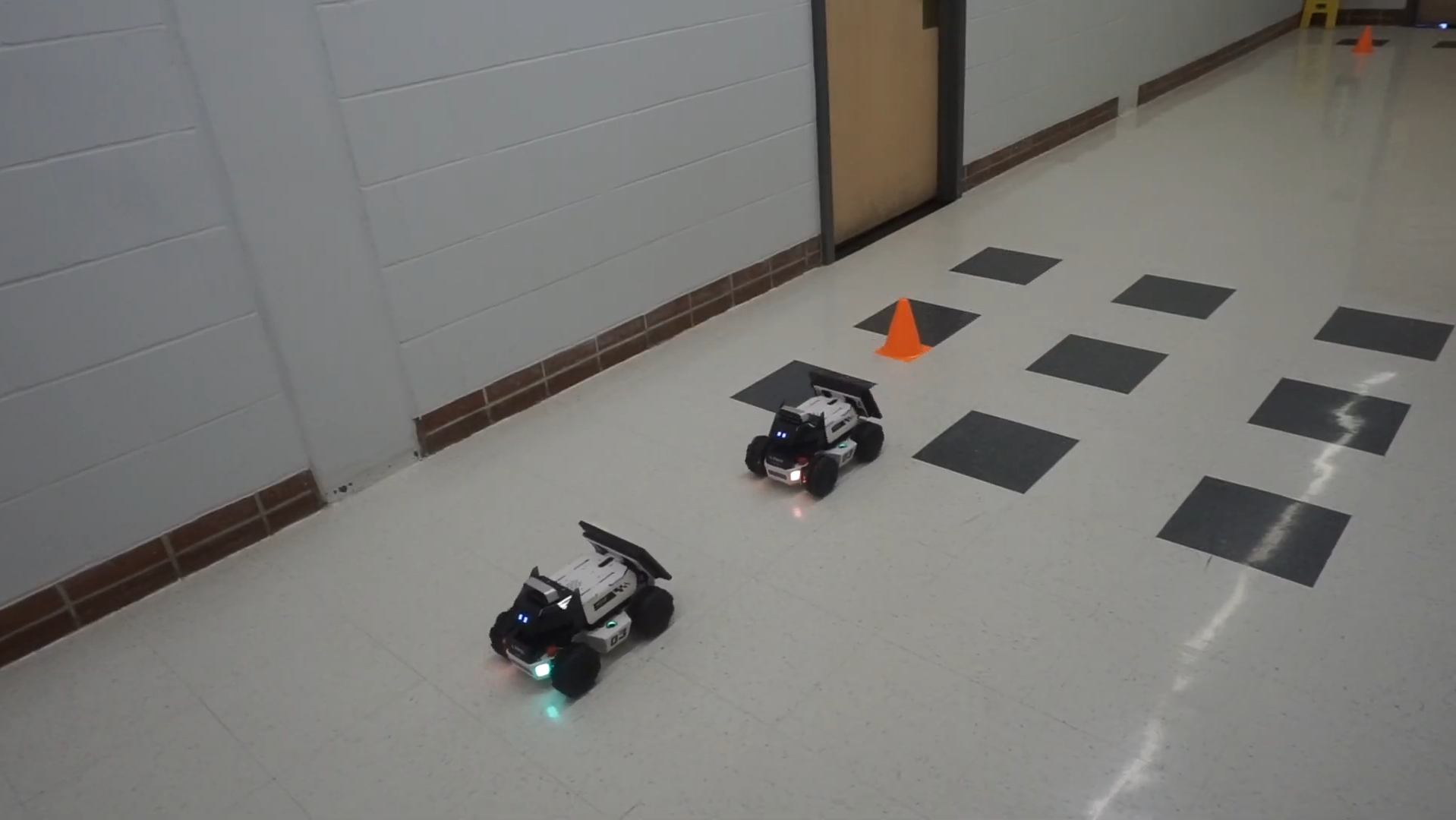}
        \caption{t = $\SI{6}{s}$}
        \label{fig:second}
    \end{subfigure}
    \hfill
    \begin{subfigure}[t]{0.32\textwidth}
        \centering
        \includegraphics[width=\textwidth,keepaspectratio]{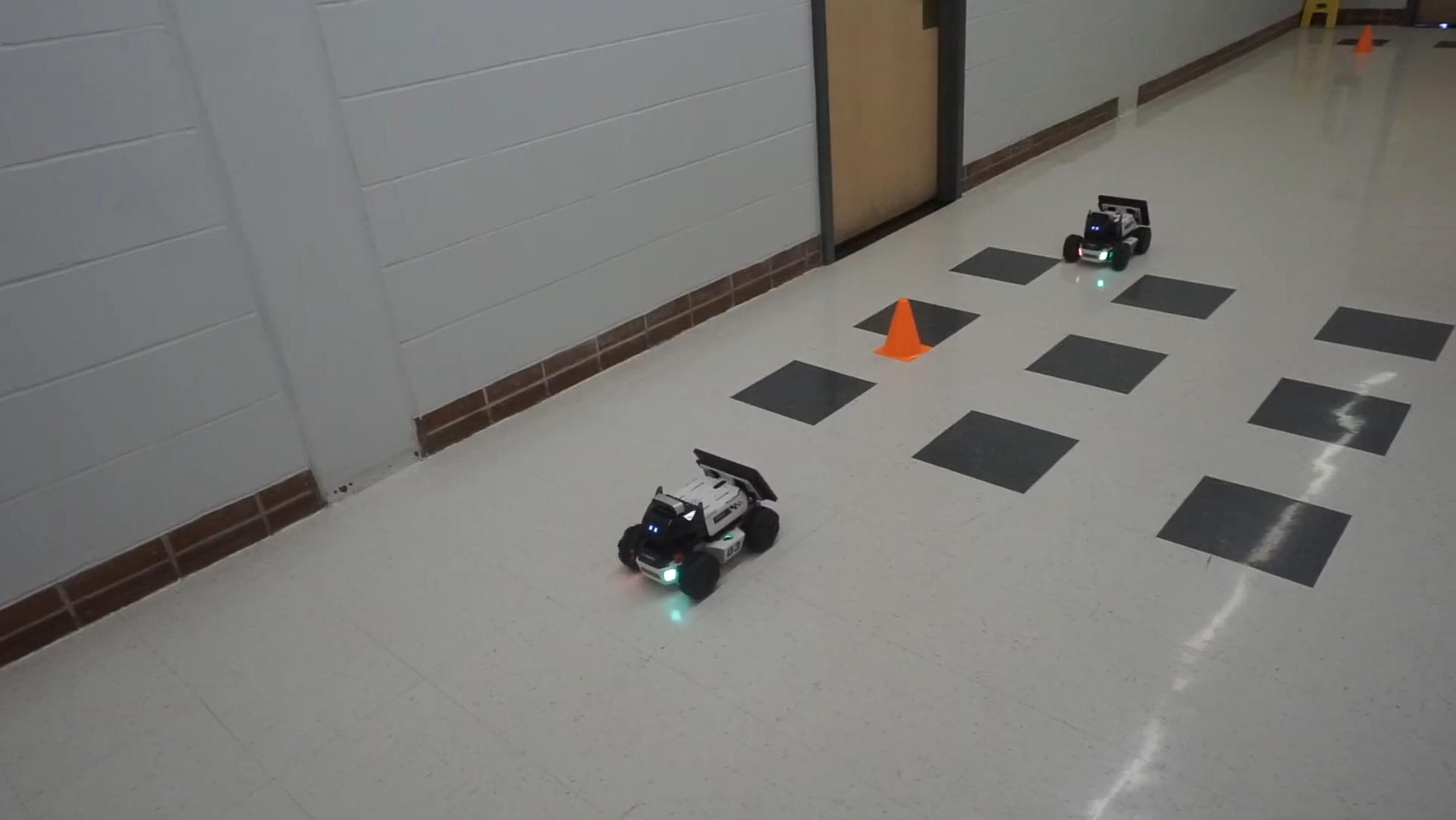}
        \caption{t = $\SI{1}{s}$}
        \label{fig:third}
    \end{subfigure}

    \caption{Position trajectories (Figure~\ref{fig:state_equiv}) and snapshots during one run (Figures~\ref{fig:second} and \ref{fig:third}).}
    \label{fig:main}
\end{figure*}
\subsubsection{Results}
In this section, we present the experimental results of the cruise control application on the \limo\, robots, which illustrate the equivalence between the optimal control strategies derived in the previous section.\\
The initial state of the ego \limo\, (that we control) is $x_0 = [0.0\,\, 0.5]^{\mathrm{T}}$. The desired reference velocity is $v_{\mathrm{ref}} = \SI{0.6}{m/s}$. The plant dynamics \eqref{eq:cc-plant-dynamics} use $\tau = 0.1$ and $k = 1.4$, whereas the model-based controller uses the estimates $\hat{\tau} = 0.3$ and $\hat{k} = 1.2$. The front \limo\, car cruises at a constant speed of \SI{0.1}{m/s} starting from $p_{\mathrm{f}}=
\SI{4.0}{m}$. The constraint parameters are $\delta = \SI{1}{m}$, $\xi = 1$ while the admissible control set is $\mathcal{U} = [0.1\,\,0.4]$. The cost parameters are $q = h = 1$ and $r = 0.5$. The zero order hold sampling time is selected to \SI{0.1}{s} for a control horizon of $T = \SI{15}{s}$.

The framework is implemented in the \rostwo\, system \cite{ros2}. The main controller node implements the model-based penalized control strategy \eqref{eq:cc-model-final-control} and the original optimal control strategy \eqref{eq:cc-plant-final-control}, depending on the desired mode of operation. To implement these strategies, real-time state feedback of the ego \limo\, is required as well as an estimate of the front \limo's position and velocity. For the former, an extended Kalman filter is utilized from the robot localization package \cite{moore2014}. To achieve the latter, in practice, real ACC systems leverage mainly radar measurements. However, since the \limo\, robots are not equipped with this type of sensor, we implement direct communication of the front \limo's position and velocity to the controller of the ego \limo. 

We conduct two runs of the cruise control example, one for the model-based penalized control and one for the original optimal control, and compare the results. The video of these runs is available at \url{https://www.youtube.com/watch?v=pMSZKlU5O44}. Figure~\ref{fig:main} depicts the position trajectories of the \limo\, cars in both runs, as well as specific snapshots during one run. $p_{\mathrm{mb}}$ and $p_{\mathrm{opt}}$ are the position trajectories of the ego \limo\, as a result of the model-based control strategy \eqref{eq:cc-model-final-control} and the original optimal control strategy \eqref{eq:cc-plant-final-control}, respectively. Figure~\ref{fig:control-inputs} depicts the control input trajectories generated by \eqref{eq:cc-model-final-control} and \eqref{eq:cc-plant-final-control}. These figures illustrate that, despite model-mismatch, the model-based penalized control strategy can recover the optimal control strategy. This is because the equivalence of optimal control trajectories follows from the equivalence of the constrained Hamiltonian minimizers, not from equality of the dynamics. 

\begin{figure}[t]
    \centering
    \resizebox{\columnwidth}{!}{
%
%
\begin{tikzpicture}

\begin{axis}[%
width=0.7\columnwidth,
height=0.3\columnwidth,
at={(1.454in,0.741in)},
scale only axis,
separate axis lines,
every outer x axis line/.append style={black},
every x tick label/.append style={font=\color{black}},
every x tick/.append style={black},
xmin=0,
xmax=15,
xlabel={Time},
every outer y axis line/.append style={black},
every y tick label/.append style={font=\color{black}},
every y tick/.append style={black},
ymin=0,
ymax=0.5,
ylabel={Control Input},
axis background/.style={fill=white},
xmajorgrids,
ymajorgrids,
legend style={legend cell align=left, align=left}
]
\addplot [color=green, line width=1.2pt]
  table[row sep=crcr]{%
0.0567159429999995	0.4\\
0.156767200999999	0.4\\
0.256679667999999	0.4\\
0.356776164999999	0.4\\
0.456831658	0.4\\
0.556898030999999	0.4\\
0.65681984	0.4\\
0.756937549	0.4\\
0.856937553999999	0.4\\
0.956819895	0.4\\
1.056879816	0.4\\
1.157131035	0.4\\
1.256577611	0.4\\
1.356641723	0.4\\
1.456532186	0.4\\
1.556758173	0.4\\
1.656777172	0.4\\
1.756718226	0.4\\
1.856829413	0.4\\
1.956898496	0.4\\
2.056760331	0.4\\
2.156717644	0.4\\
2.256714315	0.4\\
2.356860949	0.4\\
2.456759969	0.4\\
2.556855016	0.4\\
2.656921098	0.4\\
2.756781251	0.4\\
2.856783169	0.4\\
2.956616731	0.4\\
3.056605564	0.4\\
3.156872107	0.4\\
3.25681276	0.4\\
3.35690241	0.4\\
3.456690505	0.4\\
3.55672255	0.4\\
3.656787081	0.4\\
3.756617802	0.4\\
3.856594932	0.4\\
3.956924395	0.4\\
4.056911246	0.4\\
4.156695333	0.4\\
4.256695542	0.4\\
4.356762311	0.4\\
4.456790405	0.4\\
4.556656651	0.4\\
4.656669797	0.4\\
4.756707509	0.4\\
4.856761813	0.4\\
4.956689706	0.4\\
5.056694696	0.4\\
5.156682536	0.4\\
5.256693697	0.4\\
5.356635034	0.4\\
5.456818987	0.4\\
5.556955667	0.4\\
5.656855499	0.4\\
5.756538303	0.4\\
5.856888225	0.4\\
5.956640744	0.4\\
6.056676781	0.4\\
6.156530464	0.4\\
6.256920507	0.4\\
6.356699885	0.4\\
6.456644792	0.4\\
6.556591222	0.1\\
6.656714454	0.1\\
6.756654753	0.1\\
6.856760275	0.1\\
6.956480539	0.1\\
7.056817899	0.1\\
7.156586259	0.1\\
7.256631284	0.1\\
7.356643985	0.1\\
7.456920414	0.1\\
7.556573732	0.1\\
7.65638559	0.1\\
7.756510694	0.1\\
7.856726876	0.1\\
7.956743747	0.1\\
8.056598412	0.1\\
8.156732777	0.1\\
8.25683352	0.1\\
8.356747383	0.1\\
8.456536963	0.1\\
8.556658125	0.1\\
8.656898884	0.1\\
8.757095911	0.1\\
8.856647242	0.1\\
8.956545626	0.1\\
9.056722127	0.1\\
9.156620682	0.1\\
9.256775178	0.1\\
9.356463025	0.1\\
9.456671511	0.1\\
9.556702237	0.1\\
9.65662089	0.1\\
9.756691696	0.1\\
9.856584162	0.1\\
9.956685467	0.1\\
10.056507205	0.1\\
10.156652222	0.1\\
10.256483017	0.1\\
10.356660792	0.1\\
10.456656838	0.1\\
10.556606994	0.1\\
10.656627813	0.1\\
10.757038121	0.1\\
10.856671899	0.1\\
10.956594471	0.1\\
11.056566509	0.1\\
11.156626427	0.1\\
11.256699592	0.1\\
11.356645264	0.1\\
11.456641694	0.1\\
11.556573634	0.1\\
11.656617293	0.1\\
11.756459271	0.1\\
11.856630237	0.1\\
11.956546342	0.1\\
12.056703487	0.1\\
12.156546864	0.1\\
12.256628532	0.1\\
12.356680021	0.1\\
12.456586161	0.1\\
12.556795791	0.1\\
12.656707173	0.1\\
12.756639404	0.1\\
12.856531575	0.1\\
12.956492763	0.1\\
13.056536627	0.1\\
13.156699296	0.1\\
13.256760165	0.1\\
13.356954769	0.1\\
13.45671183	0.1\\
13.556818394	0.1\\
13.656761659	0.1\\
13.756422484	0.1\\
13.85658594	0.1\\
13.956719564	0.1\\
14.056803194	0.1\\
14.156505096	0.1\\
14.256698146	0.1\\
14.356549821	0.1\\
14.456691733	0.1\\
14.556482199	0.1\\
14.656753228	0.1\\
14.756690018	0.1\\
14.85650518	0.1\\
14.956650175	0.1\\
};
\addlegendentry{$u_{\mathrm{opt}}$}

\addplot [color=blue, dashed, line width=1.2pt]
  table[row sep=crcr]{%
0.193628069	0.4\\
0.294033884	0.4\\
0.394080308	0.4\\
0.493701697	0.4\\
0.593877204	0.4\\
0.693634531	0.4\\
0.793856419	0.4\\
0.894051092	0.4\\
0.99402973	0.4\\
1.093763037	0.4\\
1.19379238	0.4\\
1.29402982	0.4\\
1.393796945	0.4\\
1.493863516	0.4\\
1.593743597	0.4\\
1.693968622	0.4\\
1.79387948	0.4\\
1.893699253	0.4\\
1.99386677	0.4\\
2.094163755	0.4\\
2.193631065	0.4\\
2.293836155	0.4\\
2.393747553	0.4\\
2.494014343	0.4\\
2.593807481	0.4\\
2.693888434	0.4\\
2.793863187	0.4\\
2.893892398	0.4\\
2.994005962	0.4\\
3.093968745	0.4\\
3.193963484	0.4\\
3.293910865	0.4\\
3.393974931	0.4\\
3.493999838	0.4\\
3.593890245	0.4\\
3.694063549	0.4\\
3.79390909	0.4\\
3.893658307	0.4\\
3.993726545	0.4\\
4.093824753	0.4\\
4.193906161	0.4\\
4.293814857	0.4\\
4.393820406	0.4\\
4.493838078	0.4\\
4.593767148	0.4\\
4.693722183	0.4\\
4.793763979	0.4\\
4.89391401	0.4\\
4.993720147	0.4\\
5.093784048	0.4\\
5.193861043	0.4\\
5.29399545	0.4\\
5.394009809	0.4\\
5.493643997	0.4\\
5.594021164	0.4\\
5.694021474	0.4\\
5.793651815	0.4\\
5.893638912	0.4\\
5.994092578	0.4\\
6.093891082	0.4\\
6.193713524	0.4\\
6.293758756	0.4\\
6.393786024	0.4\\
6.493856125	0.4\\
6.593823585	0.4\\
6.693809681	0.1\\
6.793894997	0.1\\
6.893840476	0.1\\
6.993817654	0.1\\
7.093673686	0.1\\
7.193745502	0.1\\
7.29369956	0.1\\
7.393900784	0.1\\
7.493892453	0.1\\
7.593871319	0.1\\
7.69388474	0.1\\
7.793838751	0.1\\
7.893652568	0.1\\
7.993935865	0.1\\
8.093744513	0.1\\
8.193564526	0.1\\
8.293685933	0.1\\
8.393752848	0.1\\
8.493669327	0.1\\
8.593567728	0.1\\
8.693424179	0.1\\
8.793842321	0.1\\
8.893721873	0.1\\
8.993761705	0.1\\
9.093762453	0.1\\
9.193796377	0.1\\
9.2936886	0.1\\
9.39360057	0.1\\
9.493688433	0.1\\
9.593908579	0.1\\
9.693520149	0.1\\
9.793514079	0.1\\
9.893519534	0.1\\
9.993843044	0.1\\
10.093679322	0.1\\
10.193774239	0.1\\
10.293434526	0.1\\
10.393754529	0.1\\
10.493919201	0.1\\
10.593911677	0.1\\
10.693743107	0.1\\
10.794007242	0.1\\
10.894003336	0.1\\
10.993964441	0.1\\
11.093674001	0.1\\
11.193845222	0.1\\
11.293732723	0.1\\
11.393568413	0.1\\
11.493558432	0.1\\
11.593598544	0.1\\
11.69384204	0.1\\
11.793791785	0.1\\
11.893654172	0.1\\
11.993837709	0.1\\
12.093913251	0.1\\
12.193790332	0.1\\
12.293753856	0.1\\
12.393795582	0.1\\
12.49403372	0.1\\
12.593649844	0.1\\
12.693573856	0.1\\
12.793762509	0.1\\
12.893759096	0.1\\
12.993831789	0.1\\
13.093652358	0.1\\
13.193851584	0.1\\
13.293722109	0.1\\
13.393889711	0.1\\
13.493932827	0.1\\
13.593902231	0.1\\
13.693653412	0.1\\
13.793725761	0.1\\
13.893666639	0.1\\
13.993791066	0.1\\
14.093778077	0.1\\
14.193846906	0.1\\
14.293892873	0.1\\
14.393719883	0.1\\
14.493775812	0.1\\
14.593889508	0.1\\
14.693786604	0.1\\
14.793792785	0.1\\
14.893744307	0.1\\
14.993697049	0.1\\
};
\addlegendentry{$u_{\mathrm{mb}}$}

\end{axis}
\end{tikzpicture}%
}
    \caption{Equivalence of control inputs}
    \label{fig:control-inputs}
\end{figure}

\subsection{Online Learning Cruise Control}\label{sec:learning-cruise-control}
Now we apply the Hamiltonian gradient learning framework developed in Section~\ref{sec:hamiltonian-gradient-learning} to the unconstrained cruise control problem. The goal is to learn the plant Hamiltonian gradient from online data and use it to improve the policy without resetting the physical state. The algorithm is implemented in \rostwo\ on the ego \limo. To make the
physical robot follow the plant dynamics in \eqref{eq:hgl-dynamics},
the raw velocity control input produced by the learning policy is passed through a discrete PT1 filter before it is commanded to the robot:
\begin{equation}\label{eq:hgl-pt1-discrete}
    v_{k+1}=\alpha v_k+k(1-\alpha)u_k,\qquad
    \alpha=e^{-T_s/\tau},
\end{equation}
where $v_{k+1}$ is the filtered velocity commanded to the robot, $u_k$ is the raw control input of the learning algorithm, and $T_s$ is the zero-order hold period. The validation model is chosen as the known PT1 representation of this command path.
The major parameters are $\tau=\SI{0.5}{s}$, $k=1.2$, $q=2$, $r=0.5$, $\rho=0.3$, $T_s=\SI{0.01}{s}$, and $\Delta=\SI{0.10}{s}$. With this parametrization, $\gamma_1^* = 1.30$ and when $v_{\mathrm{ref}}=\SI{0.6}{m/s}$, $\gamma_0^*=-0.29$ and 
\begin{equation*}
    \mu_{0.6}^*(x)=-1.30x+0.29.
\end{equation*}
For $v_{\mathrm{ref}}=\SI{0.0}{m/s}$, $\gamma_0^*=0$ and
\begin{equation*}
    \mu_{0.0}^*(x)=-1.30x.
\end{equation*}
The robot executes four consecutive stages without resetting its velocity. 
In Stage~1, $v_{\mathrm{ref}}=\SI{0.6}{m/s}$ and the controller learns with
exploration before exploiting the converged policy. In Stage~2,
$v_{\mathrm{ref}}=\SI{0.0}{m/s}$ and the learn--then--exploit procedure is
repeated. In Stage~3, the reference returns to
$\SI{0.6}{m/s}$ and the controller exploits the converged policy stored in Stage~1. Finally, Stage~4 sets the reference to $\SI{0.0}{m/s}$ and the policy stored in Stage~2 is exploited. This schedule tests online acquisition and subsequent recall of reference-dependent affine policies on one continuous plant trajectory.

\begin{figure}[t]
    \centering
    \setlength{\figurewidth}{0.76\columnwidth}
    \setlength{\figureheight}{0.5\columnwidth}
    \input{figures/ground_truth_velocity.tex}
    \caption{Four-stage velocity tracking and reference schedule.}
    \label{fig:velocity-tracking-reference}
\end{figure}

\begin{figure}[t]
    \centering
    \setlength{\figurewidth}{0.8\columnwidth}
    \setlength{\figureheight}{0.5\columnwidth}
    \input{figures/hamiltonian_gradients-zoom.tex}
    \caption{Hamiltonian-gradient reconstruction}
    \label{fig:hamiltonian-gradients}
\end{figure}

\begin{figure}[t]
    \centering
    \setlength{\figurewidth}{0.76\columnwidth}
    \setlength{\figureheight}{0.5\columnwidth}
    \input{figures/policy_coefficients-final.tex}
    \caption{Policy coefficients' evolution and convergence.}
    \label{fig:policy-coefficients}
\end{figure}

Figure~\ref{fig:velocity-tracking-reference} shows that the \limo\
successfully tracks both reference velocities throughout the four-stage
experiment. The small steady-state tracking offset when
$v_{\mathrm{ref}}=\SI{0.6}{m/s}$ is a consequence of the cost definition,
not a failure of the learning algorithm. In particular, maintaining a
nonzero velocity requires a nonzero steady-state input, and the term $ru^2$
penalizes that input. The optimal policy therefore balances tracking accuracy
against control effort and can intentionally settle slightly below the
reference.

As shown in Figure~\ref{fig:hamiltonian-gradients}, the learned plant
Hamiltonian gradient converges to the model/reference gradient. Once learning
has converged and exploration has been removed, the two gradients match and
evaluate to zero along the converged policy, as required by the Hamiltonian
stationarity condition. Figure~\ref{fig:policy-coefficients} shows that the
learned affine-policy coefficients, $k_1$ and $k_0$, converge to their corresponding optimal
values, $k_1^*$ and $k_0^*$, at the same time that the Hamiltonian gradients match. This simultaneous convergence experimentally validates the matching conditions of
Theorem~\ref{theorem:two} for the smooth, unconstrained Hamiltonians of the
scalar affine system with quadratic cost.

The implementation code and a video of the complete robotic experiment
are available at
\url{https://github.com/Panos20102k/Multi-Limo-Control}
and
\url{https://youtu.be/p3qw2oOFOpg},
respectively.

\section{Concluding Remarks}\label{sec:conclusions}
We studied continuous optimal control with safety constraints, unknown plant dynamics, and an available approximate model. By applying PMP to a penalized model-based problem, we established structural conditions under which its constrained Hamiltonian minimizer coincides with that of the plant. The penalty that aligns the model and plant affects the state and costate evolution but does not explicitly enter the pointwise minimization with respect to the control. Consequently, control optimality depends on matching the relevant Hamiltonian structure, not on reproducing every component of the plant dynamics.

We then turned this structural observation into a learning objective. For smooth scalar affine systems with quadratic costs, the proposed integral policy-evaluation equation identifies directly the dynamics-dependent factors in the plant Hamiltonian gradient. Recursive least squares and policy improvement operate online along one continuous physical trajectory, without complete model identification or episodic state resets. The robotic experiments demonstrated both model--plant control equivalence and convergence of the learned Hamiltonian gradient and affine policy to their optimal counterparts.

These results support a task-oriented view of learning for control: the learner should recover the smallest plant quantity needed by the optimality conditions. Extending Hamiltonian-gradient learning to active safety constraints, nonlinear and multi-input systems, noisy measurements, and stochastic dynamics remains an important direction for future work.

\balance

\bibliographystyle{IEEEtran}
\bibliography{IEEEabrv,literature_panos,IDS_Publications_03032026}

\end{document}